\documentclass[conference]{IEEEtran}
\usepackage{amsfonts,amsmath,amsthm,url}

\newtheorem{thm}{Theorem}
\newtheorem{cor}[thm]{Corollary}
\newtheorem{prop}[thm]{Proposition}
\newtheorem{lem}[thm]{Lemma}

\theoremstyle{definition}
\newtheorem{defn}[thm]{Definition}

\theoremstyle{remark}

\title{On the algebraic representation\\ 
  of selected optimal non-linear binary codes}

\author{%
  \IEEEauthorblockN{Marcus Greferath and Jens Zumbr\"agel}
  \IEEEauthorblockA{%
    Claude Shannon Institute\\
    University College Dublin\\
    Belfield, Dublin 4, Ireland\\
    Email: \{marcus.greferath, jens.zumbragel\}@ucd.ie}}

\newcommand{\F}{{\mathbb F}}
\newcommand{\Z}{{\mathbb Z}}
\newcommand{\rmod}{\ \ {\rm mod}\ }
\newcommand{\smod}{\ {\rm mod}\ }
\newcommand{\om}{\mbox{\tiny$\omega$}}

\begin{document}

\maketitle

\begin{abstract}
  \boldmath Revisiting an approach by Conway and Sloane we investigate
  a collection of optimal non-linear binary codes and represent them
  as (non-linear) codes over $\Z_4$. The Fourier transform will be
  used in order to analyze these codes, which leads to a new algebraic
  representation involving subgroups of the group of units in a
  certain ring.

  One of our results is a new representation of Best's $(10, 40, 4)$
  code as a coset of a subgroup in the group of invertible elements of
  the group ring $\Z_4[\Z_5]$. This yields a particularly simple
  algebraic decoding algorithm for this code.

  The technique at hand is further applied to analyze Julin's $(12,
  144, 4)$ code and the $(12, 24, 12)$ Hadamard code.  It can also be
  used in order to construct a (non-optimal) binary $(14, 56, 6)$
  code.
\end{abstract} 


\section{Introduction}

Let $R$ be a finite ring and let $\delta$ be a metric on $R$, which is
additively extended to a metric on $R^n$.  An $(n, M, d)_{R,\delta}$
code is a (not necessarily linear) subset of~$R^n$ having
cardinality~$M$ and minimum distance $d$ with respect to the
metric~$\delta$.  Define $A_{R,\delta}(n, d)$ to be the maximum number
$M$ such that an $(n, M, d)_{R,\delta}$ code over $R$ exists.  A
fundamental task in coding theory is to determine for given $n$ and
$d$ the number $A_{R,\delta}(n, d)$, and to find a corresponding {\em
  optimal} code.  In the following we consider the case $R=\F_2$ and
$\delta=\delta_H$, the Hamming metric, or $R=\Z_4$ and
$\delta=\delta_{\rm Lee}$, the Lee metric.  We will use the simplified
notations $(n, M, d)$ for $(n, M, d)_{\F_2,\delta_H}$ and $(n, M,
d)_{\Z_4}$ for $(n, M, d)_{\Z_4,\delta_{\rm Lee}}$.

The purpose of this paper is to provide further insight into the
algebraic representation of some optimal non-linear binary codes.  Our
work is based on the quaternary constructions for the Best code and
the Julin code given by Conway and Sloane, in fact, we take the
algebraic representation in~\cite{cons} a step further in that we
involve the multiplicative structure of the ambient group ring.  More
specifically, we use the Fourier transform as well as subgroups of the
unit group of a group ring to analyze the codes.

As a result, we obtain a new succint description of Best's $(10, 40,
4)$ code as a subgroup coset in the unit group of the ring.  This
representation gives rise to a novel algebraic decoding algorithm for
this code.  We also apply this technique to analyze Julin's $(12, 144,
4)$ code and the $(12, 24, 12)$ Hadamard code.  We further use it to
construct a (non-optimal) binary $(14, 56, 6)$ code.


\section{Preliminaries}

\subsection{Codes over $\Z_4$}

In the beginning of the 90s several optimal non-linear binary codes
were recognized as Gray images of $\Z_4$-linear codes (see, e.g.,
\cite{hammons}).  Since then codes over $\Z_4$ or over more general
rings gained much attention in the literature.

Define the Lee weight on $\Z_4$ by
\[ w_{\rm Lee}: \Z_4 \to {\mathbb N}\,,\quad
x  \mapsto \min\{ x, 4-x\} \:. \]
Setting $\delta_{\rm Lee}(x,y)=w_{\rm Lee}(x-y)$, it turns out that
$(\Z_4,\delta_{\rm Lee})$ is isometric to $(\Z_2^2,\delta_H)$
via the {\em Gray mapping}
\[\Z_4 \to \Z_2^2 \,,\quad a+2\,b \,\mapsto\, a\,(0,1) + b\,(1,1) \:, \] 
where $a,b\in\{0,1\}$.  The componentwise extension of this mapping to
$\Z_4^n$ yields an isometry between $(\Z_4^n,\delta_{\rm Lee})$ and
$(\Z_2^{2n},\delta_H)$.

We note that in order to view a (non-linear) binary code as a linear
code over $\Z_4$ it is necessary that its cardinality is a power
of~$2$.  Here we are interested in codes that do not satisfy this
condition.

\subsection{Group rings and the Fourier transform}

A natural algebraic framework for cyclic codes is provided by group
rings.  For a finite ring $R$ the {\em group ring} $R[\Z_n]$ is the
set of all $R$-valued functions on $\Z_n$, equipped with the natural
addition and the multiplication $\star$ that is given by the cyclic
convolution, i.e., \[ (f\star g)\, (j) \; := \; \sum_{i\in\Z_n}
f(i)\,g(j-i) \:, \] where $f,g\in R[\Z_n]$ and $j\in\Z_n$.

Letting $\delta_j$, for $j\in\Z_n$, be the element in $R[\Z_n]$
defined by $\delta_j(j)=1$ and $\delta_j(i)=0$ if $i\ne j$, we can
write any element $f\in R[\Z_n]$ as a sum
$\sum_{i\in\Z_n}f(i)\delta_i$.  The element $f$ will also be denoted
sometimes by a vector $(f(0),f(1),\dots,f(n\!-\!1))$.

Using the concept of a group ring, a cyclic code of length~$n$ over
the ring~$R$ can be understood as a subset in $R[\Z_n]$ that is closed
under multiplication by the element $\delta_1$.

Now we define the discrete Fourier transform and state its fundamental
properties (cf.\ \cite{dv}).

\begin{defn}
  Let $S\supseteq R$ be a ring extension that contains a primitive
  $n$-th root of unity $\omega$.  For $f\in S[\Z_n]$ define the {\em
    Fourier transform\/} $\hat{f} \in S[\Z_n]$ by \[ \hat{f}(i) :=
  \sum_{j\in \Z_n} f(j)\,\omega^{-ji} \:. \]
\end{defn}

\begin{prop}
  For $f, g\in S[\Z_n]$ there holds
  \[ (f \star g)\hat{} = \hat{f} \cdot \hat{g} \:, \]
  where $\cdot$ denotes the componentwise product.
\end{prop}

\begin{prop}
  Assume that $n\in R^\times$.  Define the inverse transform by \[
  \tilde{f}(i) := \frac{1}{n} \sum_{j \in \Z_n}
  f(j)\,\omega^{ij} \:. \] Then it holds $\tilde{\hat{f}}=f =
  \hat{\tilde{f}}$ for all $f\in S[\Z_n]$.
\end{prop}

We will further use the following results.

\begin{lem}\label{lemgr}
  For $f, g\in R[\Z_n]$ there holds 
  \[ \sum_{j\in\Z_n} (f\star g)(j) = \Big( \sum_{i\in\Z_n} f(i) \Big) \, \Big(
  \sum_{j\in\Z_n} g(j) \Big) \]
\end{lem}

\begin{proof}
  We have $\sum_j (f\star g)\, (j) = \sum_j \sum_i f(i)\,g(j-i) =
  \sum_i \sum_j f(i)\,g(j) = \big( \sum_i f(i) \big) \, \big( \sum_j g(j)
  \big)$.
\end{proof}\smallskip


\begin{cor}\label{corgr}
  If $f\in R[\Z_n]$ is invertible then $\sum_i f(i)\in R^\times$.
\end{cor}

\begin{proof}
  If there exists $g\in R[\Z_n]$ with $f\star g=1$ then
  $1 = \sum_j(f\star g)(j) = \big( \sum_i f(i) \big) \, \big( \sum_j g(j)
  \big)$, so that $\sum_i f(i)$ is invertible.
\end{proof}


\section{The (10,40,4) Best code}

\subsection{The original binary code}

It is possible to prove $A_{\F_2,\delta_H}(10,4)\le 40$ by a
modification of the linear programming bound (see, e.g.,
\cite[p.~541]{mws}).  Best~\cite{best} came up with the construction
of a binary code meeting this bound.  This optimal binary $(10,40,4)$
code consists of the words
\[ 0100000011,\ 0011111101,\ 1100101100,\ 0001010111 \] together with
all cyclic shifts of these.  The distance enumerator of Best's code is
given by \[ D_H(x,y) = x^{10} + 22\, x^6y^4 + 12\, x^4y^6 + 5\, x^2y^8
\:. \] Its automorphism group is a semidirect product of the dihedral
group $D_{5}$ and $\Z_2^5$ and hence has $320$ elements.  Litsyn and
Vardy~\cite{litvar} showed that Best's code is unique, i.e., any
binary $(10,40,4)$ code must be isometric to Best's code.
Furthermore, it can be shown (cf.\ \cite{best, cons}) that applying
the so-called Construction~A to Best's $(10,40,4)$ code yields the
densest sphere packing known in $10$ dimensions.

\subsection{The pentacode over $\Z_4$}

The $(10, 40, 4)$ Best code cannot be recognized as a submodule over
$\mathbb{Z}_4$, since the cardinality is not a power of~$2$.
Nevertheless, it was observed by Conway and Sloane~\cite{cons} that
the Best code has a sensible interpretation as a code over $\Z_4$.
Namely, define a code $P\subseteq \Z_4^5$ consisting of all words \[
(c-d,b,c,d,b+c) \quad\text{where }b,c,d \in \{1,3\}\] and all cyclic
shifts of these.  This code has parameters $(5,40,4)_{\Z_4}$ and is
called the {\em pentacode}.  It can be shown that the Gray image of
$P$ is (up to equivalence) the $(10,40,4)$ code discovered by Best.
Furthermore, the code~$P$ is invariant under the automorphisms
\begin{align*}
  (a,b,c,d,e) &\;\mapsto\; (-a, -b, -c, -d, -e) \:,\\
  (a,b,c,d,e) &\;\mapsto\; (-a, 2-b, c, 2-d, -e) \:,\\
  (a,b,c,d,e) &\;\mapsto\; (b,c,d,e,a) \:,\\
  (a,b,c,d,e) &\;\mapsto\; (2+e,2+d,2+c,2+b,2+a) \:.
\end{align*}

\subsection{Spectral analysis of the pentacode}

In order to apply the Fourier transform, we chose to analyze the
pentacode, because Best's original binary code does not satisfy $n \in
\F_2^\times$.  For this we find that the Galois ring ${\rm GR}(4,4)$
as an extension of $\Z_4$ contains the required primitive $5$-th root
of unity $\omega$.  The minimal polynomial of $\omega$ over $\Z_4$ is
given by \[ \varphi_\omega \; = \; x^4+x^3+x^2+x+1 \:. \]

\begin{table}
  \caption{The Fourier transform of the pentacode}
  \label{tabspec}

  The Fourier transform of the pentacode is given by the following 
  vectors together with their negations.
  {\footnotesize \[\!\!\!\!\begin{array}{l}
      (\,1\,,\, 3\om^3\!+\!3\om^2\!+\!3\om\!+\!2\,,\, \om^3\!+\!3\,,\, 
      \om^2\!+\!3\,,\, \om\!+\!3\,)\\
      (\,1\,,\, 3\om\!+\!1\,,\, 3\om^2\!+\!1\,,\, 
      3\om^3\!+\!1\,,\, \om^3\!+\!\om^2\!+\!\om\!+\!2\,)\\
      (\,1\,,\, 3\om^2\!+\!\om\,,\, \om^3\!+\!2\om^2\!+\!\om\!+\!1\,,\, 
      \om^3\!+\!3\om\,,\, 2\om^3\!+\!3\om^2\!+\!3\om\!+\!3\,)\\
      (\,1\,,\, 3\om^3\!+\!\om^2\,,\, 3\om^3\!+\!3\om^2\!+\!2\om\!+\!3\,,\, 
      \om^3\!+\!\om^2\!+\!2\om\!+\!1\,,\, \om^3\!+\!3\om^2\,)\\
      (\,1\,,\, 2\om^3\!+\!\om^2\!+\!\om\!+\!1\,,\, 3\om^3\!+\!\om\,,\, 
      3\om^3\!+\!2\om^2\!+\!3\om\!+\!3\,,\, \om^2\!+\!3\om\,)\\
      (\,1\,,\, \om^3\!+\!3\om^2\!+\!3\om\,,\, \om^3\!+\!2\om\!+\!1\,,\, 
      2\om^3\!+\!3\om^2\!+\!2\om\!+\!3\,,\, 2\om^2\!+\!\om\!+\!1\,)\\
      (\,1\,,\, 2\om^3\!+\!2\om^2\!+\!3\om\!+\!3\,,\, 2\om^3\!+\!\om^2\!+\!1\,,\, 
      \om^3\!+\!2\om^2\!+\!1\,,\, 3\om^3\!+\!3\om^2\!+\!\om\,)\\
      (\,1\,,\, \om^2\!+\!\om\!+\!2\,,\, 3\om^3\!+\!3\om\!+\!1\,,\, 
      \om^3\!+\!\om\!+\!2\,,\, 3\om^2\!+\!3\om\!+\!1\,)\\
      (\,1\,,\, \om^3\!+\!\om^2\!+\!2\om\,,\, 3\om^3\!+\!\om^2\!+\!3\,,\, 
      \om^3\!+\!3\om^2\!+\!3\,,\, 3\om^3\!+\!3\om^2\!+\!2\om\!+\!2\,)\\
      (\,1\,,\, \om^2\!+\!3\om\!+\!3\,,\, 3\om^3\!+\!2\om^2\!+\!3\om\!+\!2\,,\, 
      3\om^3\!+\!\om\!+\!3\,,\, 2\om^3\!+\!\om^2\!+\!\om\,)\\
      (\,1\,,\, \om^3\!+\!\om^2\!+\!3\om\,,\, 3\om^3\!+\!2\om^2\!+\!3\,,\, 
      2\om^3\!+\!3\om^2\!+\!3\,,\, 2\om^3\!+\!2\om^2\!+\!\om\!+\!1\,)\\
      (\,1\,,\, 2\om^2\!+\!3\om\!+\!3\,,\, 2\om^3\!+\!\om^2\!+\!2\om\!+\!1\,,\, 
      3\om^3\!+\!2\om\!+\!3\,,\, 3\om^3\!+\!\om^2\!+\!\om\,)\\
      (\,1\,,\, 2\om^3\!+\!3\om^2\!+\!3\om\,,\, \om^3\!+\!3\om\!+\!1\,,\, 
      \om^3\!+\!2\om^2\!+\!\om\!+\!2\,,\, 3\om^2\!+\!\om\!+\!1\,)\\
      (\,1\,,\, \om^3\!+\!\om^2\!+\!2\om\!+\!2\,,\, 3\om^3\!+\!\om^2\!+\!1\,,\, 
      \om^3\!+\!3\om^2\!+\!1\,,\, 3\om^3\!+\!3\om^2\!+\!2\om\,)\\
      (\,1\,,\, \om^2\!+\!\om\!+\!3\,,\, 3\om^3\!+\!3\om\!+\!2\,,\, 
      \om^3\!+\!\om\!+\!3\,,\, 3\om^2\!+\!3\om\!+\!2\,)\\
      (\,1\,,\, 3\om^3\!+\!\om^2\!+\!3\om\!+\!2\,,\, 3\om^3\!+\!2\om^2\!+\!2\om\!+\!1\,,\, 
      \om^2\!+\!2\om\!+\!3\,,\, 2\om^3\!+\!\om\!+\!3\,)\\
      (\,1\,,\, 2\om^3\!+\!3\om\!+\!1\,,\, 3\om^2\!+\!2\om\!+\!1\,,\, 
      \om^3\!+\!2\om^2\!+\!2\om\!+\!3\,,\, \om^3\!+\!3\om^2\!+\!\om\!+\!2\,)\\
      (\,1\,,\, 2\om^3\!+\!\om^2\!+\!3\om\!+\!2\,,\, 3\om^3\!+\!2\om^2\!+\!\om\!+\!1\,,\, 
      \om^3\!+\!2\om^2\!+\!3\om\,,\, 2\om^3\!+\!3\om^2\!+\!\om\!+\!3\,)\\
      (\,1\,,\, 3\om^3\!+\!\om^2\!+\!2\,,\, 3\om^3\!+\!3\om^2\!+\!2\om\!+\!1\,,\, 
      \om^3\!+\!\om^2\!+\!2\om\!+\!3\,,\, \om^3\!+\!3\om^2\!+\!2\,)\\
      (\,1\,,\, 2\om^3\!+\!\om^2\!+\!3\om\!+\!1\,,\, 3\om^3\!+\!2\om^2\!+\!\om\,,\, 
      \om^3\!+\!2\om^2\!+\!3\om\!+\!3\,,\, 2\om^3\!+\!3\om^2\!+\!\om\!+\!2\,)\\
    \end{array}\]}
\end{table}

We computed the Fourier transform of all words of the pentacode.  The
result is shown in Table~\ref{tabspec}.  It is apparent that the spectrum
of each word of $P$ solely consists of invertible elements in ${\rm
  GR}(4,4)$.  We can further make the following observation. Let
\begin{align*}
  \hat{f} &:=
  (\,1\,,\,3\,\omega+1\,,\, 3\,\omega^2+1\,,\, 
  3\,\omega^3+1\,,\, 3\,\omega^4+1\,)\\
  \hat{g} &:=
  (\,1\,,\, 3\,\omega+2\,,\, 3\,\omega^2+2\,,\, 
  3\,\omega^3+2\,,\, 3\,\omega^4+2\,)\\
  \hat{h} &:=
  (\,1\,,\, 2\,\omega+3\,,\, 2\,\omega^2+3\,,\, 
  2\,\omega^3+3\,,\, 2\,\omega^4+3\,)
\end{align*}
Then for each word $c\in P$ there holds
\[ \hat{c} =  \hat{f} \cdot (-1)^i \cdot \hat{h}^j \cdot 
\hat{g}^k \quad \text{for some }i,j\in \Z_2,\; k\in \Z_{10} \:. \]

\subsection{Algebraic representation}

The preceeding observations lead to the following algebraic structure
of the pentacode.  By applying the inverse transform and doing some
rescaling we obtain the vectors
\begin{align*}
  f &= (1,1,1,2,0)\:,\\
  g &= (2,1,0,0,0)\:,\\
  h &= (1,2,0,0,0)\:.
\end{align*}
Here $h^2=1$ and $g$ is of order $10$.  This yields the following theorem:

\begin{thm}\label{penta}
  Let $f,g,h\in\Z_4[\Z_5]$ be as above.  For each word~$c$ of the
  pentacode~$P$ there holds \[ c = f\star (-1)^i \star h^j \star
  g^k \] for some $i,j\in \Z_2$ and $k\in \Z_{10}$.  Consequently, the
  pentacode is a coset \[ P \; = \; f \star U \:, \] where $U$ is a $40$
  element subgroup of the group of invertible elements of the group
  ring $S := \Z_4[\Z_5]$.
\end{thm}

One can prove that the ring $S=\Z_4[\Z_5]$ is isomorphic to ${\rm
  GR}(4,4)\times \Z_4$, and hence the unit group $S^\times$ is
isomorphic to $\Z_2^4\times\Z_{15}\times\Z_2 \cong
\Z_2^5\times\Z_3\times\Z_5$ (see \cite{mcd}).  From this it is easy to
see that there are ${5\choose 3}_2 = 155$ subgroups $U$ of order~$40$
in the unit group of~$S$; here ${5\choose 3}_2$ denotes the Gaussian
binomial coefficient, i.e., the number of $3$-dimensional subspaces of
$\Z_2^5$.  By a computer search we found that only~$2$ of these
subgroups yield (up to equivalence) the pentacode.  Moreover, the
pentacode occurs twice among the $12$ cosets of each of these two
subgroups.

The four algebraic representations of the pentacode as $f\star U$,
where $f\in S$ and $U$ is a subgroup of $S^\times$, are given in
Table~\ref{tabbest}.  In each case, we state four generators of the
subgroup~$U$, which have order $5$, $2$, $2$, and $2$, respectively,
and we chose the representative $f$ to be of minimal degree in the
`polynomial' representation $\sum_{i\in\Z_j} f(i) \delta_1^i$.  The
original pentacode $P$ of Theorem~\ref{penta} is found as
$f_{1,1}\star U_1$; for this, note that $g^5=(1,0,0,0,2)$ and
$g^6=(0,1,0,0,0)$.

\begin{table}
  \caption{The four representations of the pentacode as $f\star U$}
  \label{tabbest}

  {\centering
    \renewcommand{\arraystretch}{1.1}
    \begin{tabular}{r|c}
      $U_1 = \langle\,(0,1,0,0,0)\,,$~ \\
      $(3,0,0,0,0)\,,$~ & ~$f_{1,1}=(3,1,1,0,0)$ \\
      $(1,2,0,0,0)\,,$~ & ~$f_{1,2}=(3,3,1,0,0)$ \\
      $(1,0,0,0,2)\,\rangle$~ \\[1mm]\hline\\[-2mm]
 
      $U_2 = \langle\,(0,1,0,0,0)\,,$~ & \\
      $(3,0,0,0,0)\,,$~ & ~$f_{2,1}=(3,1,0,1,0)$ \\
      $(1,0,2,0,0)\,,$~ & ~$f_{2,2}=(1,3,0,1,0)$ \\
      $(1,0,0,2,0)\,\rangle$~ \\
    \end{tabular}\\
    }
\end{table}

For each subgroup $U_j$, $j=1,2$, the pentacodes $f_{j,1}\star U_j$
and $f_{j,2}\star U_j$ are related as $f_{j,2}\star U_j = a\star
f_{j,1}\star U_j$, where $a=(3,2,2,2,2)\notin U_j$ is a unit element;
this element has the property that the left multiplication \[ L_a:
S^\times \to S^\times \:, \quad x\mapsto a\star x \] is a $w_{\rm
  Lee}$-isometry on the unit group $S^\times$.

In fact, the set of $s\in S$ such that $L_s:S^\times\to S^\times$ is
an isometry of the unit group $S^\times$ equals $\langle \delta_1, -1,
a \rangle$, whereas the set of $s\in S$ such that $L_s:C\to C$ is a
code isometry equals $\langle \delta_1, -1 \rangle$, for all four
pentacodes $C$.  In particular, not all left multiplications of
elements in $U_j$ are isometries.

We note that \cite{cons} mentions also four pentacodes ${\mathcal
  B}_0$, ${\mathcal B}_1$, ${\mathcal B}_2$, ${\mathcal B}_3$, which
are constructed as orbits of a group of isometries on $\Z_4^5$.  Here,
${\mathcal B}_0$ equals $f_{1,1}\star U_1$ and ${\mathcal B}_2$ equals
$f_{1,2}\star U_1$, whereas ${\mathcal B}_1$ and ${\mathcal B}_3$ are
different from $f_{2,1}\star U_2$ and $f_{2,2}\star U_2$.

\subsection{A decoding algorithm}

We present a simple decoding algorithm for the pentacode which is
based on its algebraic representation as $f\star U$ in the group ring.
The decoding algorithm differs considerably from the one given
in~\cite{cons}.

In the following we will write $f\,g$ for the convolution $f\star g$
in a group ring.  For concreteness we consider in the ring
$S=\Z_4[\Z_5]$ the pentacode $C = f_{2,1}\,U_2 = f\,U$, with $f =
(3,3,2,1,0) = (0,1,1,1,2)^{-1} \in C$ and $U = U_2 = \langle \delta_1,
-1, g, h \rangle$, where $g=(1,0,2,0,0) = 2\delta_2+1$ and
$h=(1,0,0,2,0) = 2\delta_3+1$.  Hence $C$ is given as
\[ C = \{ f\,(-1)^i\,g^j\,h^k\,\delta_\ell\mid i,j,k\in\Z_2\,, 
\ell\in\Z_5 \} \:. \]
The following lemma gives a simple membership test for $U$.

\begin{lem}\label{lemtest}
  Let $s\in S$.  Then $s\in U$ if and only if
  \[ s \equiv \delta_\ell\rmod 2\quad\text{and}\quad
  \delta_{-\ell}\,s = 2\,(i+j\,\delta_2+k\,\delta_3)+1 \: \]
  for some $\ell\in\Z_5$ and $i,j,k\in\Z_2$.
\end{lem}

\begin{proof}
  An element $u\in U$ is of the form $(-1)^i\,g^j\,h^k\,\delta_\ell =
  (2+1)^i\,(2\delta_2+1)^j\,(2\delta_3+1)^k\,\delta_\ell =
  \big(2\,(i+j\,\delta_2+k\,\delta_3)+1\big)\delta_\ell$, for some
  $\ell\in\Z_5$ and $i,j,k\in\Z_2$.  
\end{proof}

Given $y\in S$ we can check whether $y\in C = f\,U$ by applying the
test of Lemma~\ref{lemtest} to $z=f^{-1}\,y$.  In case $y\in C$ we also
obtain this way the `message' $(i,j,k,\ell) \in \Z_2^3\times\Z_5$.

Now we present a decoding algorithm for the pentacode $C$.  Let $y\in
S$ be the received vector, and assume that $y = c+e = f\,u+e$, where
$c\in C$, $u\in U$, and $e\in S$ is an error vector of Lee weight $\le
2$.  We note that $f^{-1} = (0,1,1,1,2)$.

\begin{enumerate}
\item Compute $z := f^{-1}\,y$.

  [In fact, $z = u + f^{-1}\,e$.  Note that $u\equiv\delta_\ell\rmod
  2$, so that $w_H(u\smod 2)=1$.  Furthermore, $w_H(f^{-1}\,e\smod 2)$
  is odd if and only if $w_H(e\smod 2)$ is odd.]\smallskip

\item[2a)] Case $w_H(z\smod 2)$ odd: Check if $z\in U$ by
  applying Lemma~\ref{lemtest}.  If yes, conclude that no error
  occurred, and output the corresponding message $(i,j,k,\ell)$; if
  no, conclude that two errors occurred.

  [Note that $w_H(e\smod 2)$ is even.]\smallskip

\item[2b)] Case $w_H(z\smod 2)$ even:

  [Note that $w_H(e\smod 2)$ is odd, and hence $e=\pm\delta_r$
  for some $r\in\Z_5$.  We thus have $z = u \pm \delta_r\,f^{-1}$,
  where $u = \delta_\ell\,\big(2\,(i+j\,\delta_2+k\,\delta_3)+1\big)$
  for some $\ell\in\Z_5$ and $i,j,k\in\Z_2$.]\smallskip
  
  Determine $\ell\in\Z_5$, $r\in\Z_5$, and a sign according to the
  following table (with $m\in\Z_5$ to be chosen):\medskip

  {\centering
    \renewcommand{\arraystretch}{1.1}
    \begin{tabular}{c|c}
      $\delta_m\,z$ & $(m\!+\!\ell, m\!+\!r, \pm)$ \\[.5mm]\hline\\[-2.5mm]
      $(o,\pm 1,o,o,2)$ & $(0,0,\pm)$ \\
      $(o,o,\pm 1,o,0)$ & $(3,4,\pm)$ \\[.5mm]
      $(o,\pm 1,e,2,e)$ & $(2,4,\pm)$ \\
      $(\pm 1,o,e,0,e)$ & $(4,3,\pm)$ \\[.5mm]
      $(\pm 1,e,\pm 1,e,e)$ & $(1,4,\pm)$
    \end{tabular}\\
  }\smallskip
  Here $e$ stands for an even element and $o$ for an odd element of $\Z_4$.
  
  [Note that $f^{-1} = (0,1,1,1,2)$ and
  $2\,(i+j\,\delta_2+k\,\delta_3)+1$ is of the form $(o,0,e,e,0)$.]\smallskip

  Then $\delta_{-\ell}\, ( z \mp \delta_r\,f^{-1} ) =
  2\,(i+j\,\delta_2+k\,\delta_3)+1$ for some $i,j,k\in\Z_2$; output
  the message $(i,j,k,l)$.
\end{enumerate}

It is not hard to see that the pattern given in the table determine
uniquely $\ell\in\Z_5$, $r\in\Z_5$, and the sign of the error
$e=\pm\delta_r$.  Thus, the decoding algorithm corrects all errors of
Lee weight~$1$ and detects all errors of Lee weight~$2$.  This shows
again that the code has minimum Lee distance~$4$.


\section{Investigation of other codes}

The Gray images of $\Z_4$-codes of the form $f\star U$, where $f\in
S=\Z_4[\Z_n]$ and $U$ is a subgroup of the unit group $S^\times$,
yield good binary codes in many cases.  For example, binary codes with
parameters $(6,4,4)$, $(8,16,4)$, $(10,40,4)$, $(12,128,4)$,
$(14,392,4)$, $(12,24,6)$, and $(14,56,6)$ can be constructed this
way.  In this section we study some of these cases as well as the
$(12,144,4)$ Julin code in greater detail.  We start with a
preliminary observation.

\begin{prop}
  Let $S:=\Z_4[\Z_n]$, let $f\in S$ and $U\subseteq S^\times$, and let
  $C:=f\star U$.  Then the Lee distance between two codewords of $C$
  is even, in particular the minimum Lee distance of the code $C$ is
  even.
\end{prop}

\begin{proof}
  Let $S_{\rm odd} := \{f\in S\mid \sum_if(i)\in\Z_4^\times=\{1,3\}\}$
  and let $S_{\rm even} := S\setminus S_{\rm odd}$.  By
  Corollary~\ref{corgr}, $U\subseteq S^\times\subseteq S_{\rm odd}$.

  Case 1: $f\in S_{\rm odd}$.  Then from Lemma~\ref{lemgr} it follows
  $C = f\star U \subseteq S_{\rm odd}$.  Let $c,c'\in C\subseteq
  S_{\rm odd}$.  Then $\sum_i(c-c')(i) = \sum_ic(i) - \sum_ic'(i)$ is
  even, so that $c-c'\in S_{\rm even}$.  It follows that $w_{\rm
    Lee}(c-c') = \sum_i w_{\rm Lee}(c-c')(i)$ is even.

  Case 2: $f\in S_{\rm even}$.  Then from Lemma~\ref{lemgr} it follows
  $C = f\star U\subseteq S_{\rm even}$.  As before, for $c,c'\in C$ it
  follows $c-c'\in S_{\rm even}$ and thus $w_{\rm Lee}(c-c')$ is even.
\end{proof}

\subsection{The $(12, 144, 4)$ Julin code}

Julin's binary $(12, 144, 4)$ code~\cite{julin} (see
also~\cite[p.~70f]{mws}) can be constructed by taking as words the 132
blocks of a Steiner $S(5, 6, 12)$ system and adding six words each of
weight~$2$ and~$10$.  The code was shown to be optimal in~\cite{obk}.
Conway and Sloane~\cite{cons} provided a quaternary construction for a
canonical version of the Julin code.  In this construction one defines
the vectors
\begin{gather*}
  c_1=(0,0,1,1,2,2)\quad c_2=(0,0,2,2,1,1)\\
  c_3=(0,1,0,2,1,2)\quad c_4=(0,2,0,1,2,1)\\
  c_5=(0,1,2,0,2,1)\quad c_6=(0,1,1,3,3,2)\\
  c_7=(0,1,2,3,1,3)\quad c_8=(0,1,3,1,2,3)\\
  c_9=(0,1,3,2,3,1)\quad c_{10}=(0,2,1,1,3,3)\\
  c_{11}=(3,1,\dots,1)\ \ c_{12}=(2,0,\dots,0)\ \ c_{13}=(0,2,\dots,2)
\end{gather*}
and let the {\em quaternary Julin code} $J$ consist of all cyclic
shifts and negations of the vectors $c_1, \dots, c_{13}$.  This code
has parameters $(6, 144, 4)_{\Z_4}$.

We performed a spectral analysis of the quaternary Julin code~$J$.
Although $n=6$ is not a unit in characteristic~$4$, the Fourier
transform provided us with enough information to find the following
algebraic representation.

Let $S$ be the group ring $\Z_4[\Z_6]$, and let $g=(0,1,1,1,1,1)\in S$
which is an invertible element of order~$2$.  We consider the subgroup
$U$ of $S^\times$ generated by $-1$, $\delta_1$, and $g$, which
has~$24$ elements.

Then~$J$ is invariant under the group~$U$; it is actually a union
of~$8$ cosets of~$U$.  More concretely, $J$ can be written as
\begin{align*}
  J &= \{c_1, c_5\}\star (\{1, (0,0,3,1,0,1)\}\star U)\\
  &\quad~~~\cup\ c_5\star (\{(0,1,0,1,0,3), (1,0,0,0,1,3)\}\star U)\\
  &\quad~~~\cup\ c_1\star (0,0,0,1,1,1)\star U\\
  &\quad~~~\cup\ c_{11}\star U \:.
\end{align*}

\subsection{The $(12, 24, 6)$ Hadamard code}

The binary $(12, 24, 6)$ binary Hadamard code (see, e.g.,
\cite[p.~49]{mws}) is constructed from a Hadamard matrix~$H$ of
order~$12$, and such a matrix can be obtained from the Paley
construction.  The Hadamard matrix~$H$ is changed into a binary
matrix~$A$, where $+1$\,s are replaced by $0$\,s and $-1$\,s by
$1$\,s.  Then the {\em Hadamard code} consists of all rows of the
matrix~$A$, together with their complements.  From the definition of
Hadamard matrix the code has parameters $(12, 24, 6)$, which can be
seen to be optimal by applying the Plotkin bound
(see~\cite[p.~43]{mws}).

We now give an algebraic representation of a quaternary $(6, 24,
6)_{\Z_4}$ by using again the group ring $S=\Z_4[\Z_6]$.  Let $U$ be
the subgroup of $S^\times$ generated by the elements $(2,1,1,1,1,1)$,
$(0,3,3,3,3,3)$, and $(2,0,0,3,0,0)$ of order~$2$, as well as the
element $(0,0,1,0,0,0)$ of order~$3$.  Then $U$ has~$24$ elements, and
\[ C = f \star U \] with $f = (0,1,2,0,0,0)$ defines a $(6, 24,
6)_{\Z_4}$ code.  It can be shown that the Gray map of $C$ is a
translation of a $(12, 24, 6)$ Hadamard code.

We remark that a binary $(12, 24, 6)$ code was not found in a similar
way in the {\em binary} group ring $\Z_2[\Z_{12}]$.

\subsection{A $(14, 56, 6)$ code}

We performed a computer search to obtain an interesting heptacode.
For this we took a similar avenue as for the Best code and define:
\begin{align*}
  f &= (1,2,3,1,1,0,0)\\
  g &= (2,1,0,0,0,0,0)\\
  h &= (1,2,0,0,0,0,0)
\end{align*}
Again $h^2=1$, and here $g$ is of order $14$.  We form the code \[ H
:= \{ f\star (-1)^i\star h^j \star g^k \mid i,j\in \Z_2,\, k\in
\Z_{14}\} \:, \] which has $56$ words and is of minimum Lee distance
$6$.  The code turns out to be distance-invariant and its distance
enumerator is given by \[ D_{\rm Lee}(x,y) = x^{14} + 36\, x^8y^6 +
7\, x^6y^8 + 12\, x^4y^{10} \:. \]

\enlargethispage{-9.6cm}

The Gray image of $H$ is a binary $(14, 56, 6)$ code.  Consulting
Grassl's tables~\cite{grassl} we find that for a {\em linear\/} binary
code only the parameters $(14, 2^5, 6)$ are possible.  From Brouwer's
table~\cite{brouwer} however we find a strong competitor, namely the
doubly shortened Nordstrom-Robinson code, which is a $(14,64,6)$ code.

We remark that by shortening the code $H$ we get a code with
parameters $(13, 28, 6)$, which are the same as for the so-called
conference matrix code of length~$13$ (see~\cite[p.~57]{mws}).



\section*{Conclusions}

In this paper we provided new insights into the algebraic
representation of Best's $(10, 40, 4)$ code and other non-linear
binary codes.  We established the discrete Fourier transform as a
useful tool when analyzing non-linear codes, and found a presentation
of certain optimal non-linear binary codes as coset(s) of a subgroup
in the unit group of a ring.  It is an open problem whether
asymptotically good codes can be obtained by this method.


\section*{Acknowledgements}

The authors would like to thank Michael Kiermaier and Oliver Gnilke
for inspiring discussions.  This work was supported in part by Science
Foundation Ireland, Grant 06/MI/006 and Grant 08/IN.1/I1950.


\end{document}